\documentclass[11pt]{article}
\usepackage{amsfonts}
\usepackage{amsmath}

\newtheorem{theorem}{Theorem}

\newtheorem{lemma}[theorem]{Lemma}

\newtheorem{remark}[theorem]{Remark}

\newenvironment{proof}[1][Proof]{\noindent\textit{#1.} }{}
\input{tcilatex}
\begin{document}

\begin{center}
\textbf{Efficient estimation in the Topp-Leone distribution}\smallskip

Lazhar Benkhelifa\smallskip

(l.benkhelifa@yahoo.fr)\smallskip

\textit{Laboratory of Applied Mathematics, Mohamed Khider University,
Biskra, 07000, Algeria\smallskip }

\textit{Department of Mathematics and Informatics, Larbi Ben M'Hidi
University, Oum El Bouaghi, 04000, Algeria}\smallskip \smallskip \smallskip

\noindent \textbf{Abstract}
\end{center}

\noindent In the current paper, the estimation of the probability density
function and the cumulative distribution function of the Topp-Leone
distribution is considered. We derive the following estimators: maximum
likelihood estimator, uniformly minimum variance unbiased estimator,
percentile estimator, least squares estimator and weighted least squares
estimator. A simulation study shows that the maximum likelihood estimator is
more efficient than the others estimators.\bigskip

\medskip

\noindent \textbf{MSC class 2010 Subject Classification}: 62E15,
62N05.\smallskip 

\section{\textbf{Introduction}}

\noindent The univariate continuous Topp-Leone distribution with bounded
support was originally proposed by Topp and Leone (1955) and applied it as a
model for some failure data. The probability density function (PDF) and
cumulative distribution function (CDF) are respectively given, for $0<x<1\ $%
and $\alpha >0,$ by%
\begin{equation}
f\left( x\right) =\alpha \left( 2-2x\right) \left( 2x-x^{2}\right) ^{\alpha
-1}
\end{equation}%
and%
\begin{equation}
F\left( x\right) =\left( 2x-x^{2}\right) ^{\alpha }.
\end{equation}%
\noindent In recent years, the Topp-Leone distribution has received a huge
attention in the literature. For instance, Nadarajah and Kotz (2003) derived
the structural properties of this distribution including explicit
expressions for the moments, hazard rate function and characteristic
function. Kotz and Van Dorp (2004) proposed a generalized version of the
Topp-Leone distribution for modeling some financial data and studied its
properties. Ghitany et al. (2005) discussed some reliability measures of the
Topp-Leone distribution and their stochastic orderings. Kotz and Nadarajah
(2006) gave a bivariate generalization of this distribution. Ghitany (2007)
derived the asymptotic distribution of order statistics of this model.
Vicari et al. (2008) introduced a two-sided generalized version of the
distribution and discussed some of its properties. The moments of order
statistics from this distribution were discussed by Gen\c{c} (2012) and
MirMostafaee (2014). Gen\c{c} (2013) considered the estimation of the
stress-strength parameter for this distribution. Admissible minimax
estimates for the shape of this distribution were\textbf{\ }derived by
Bayoud (2016). Bayesian and non-Bayesian estimation of Topp-Leone
distribution based lower record values were obtained by MirMostafaee et al.
(2016).\medskip

\noindent In many situations we need to estimate the PDF, CDF or both. For
example, we use the PDF to estimate the differential entropy,
Kullback-Leibler divergence, R\'{e}nyi entropy and Fisher information; we
use the CDF to estimate the quantile function, the cumulative residual
entropy, Bonferroni and Lorenz curves whereas we can use the PDF and CDF
together for estimating the probability weighted moments, the hazard rate
function, the reverse hazard rate function and the mean deviation about
mean.\medskip

\noindent In this paper, we discuss the efficient estimation of the PDF and
the CDF for the Topp-Leone distribution by considering the following
estimators: maximum likelihood estimator (MLE), uniformly minimum variance
unbiased estimator (UMVUE), percentile estimator (PCE), least squares
estimator (LSE) and weighted least squares estimator (WLSE).\medskip

\noindent There are several similar studies for other distributions. We
mention: Pareto distribution by Dixit and Jabbari (2010, 2011), generalized
exponential-Poisson distribution by Bagheri et al. (2014), generalized
exponential distribution by Alizadeh et al.(2014), exponentiated Weibull
distribution by Alizadeh et al. (2015), Weibull extension model by Bagheri
et al. (2016), exponentiated Gumbel distribution by Bagheri et al.
(2016).\bigskip

\noindent The rest of the paper is organized as follows. In Sections 2 and
3, we derive the MLE and the UMVUE of the PDF and the CDF with their mean
squared errors (MSEs), respectively. The PCE, LSE and WLSE are considered in
Sections 4 and 5, respectively. Section 6 includes a simulation study in
order to compare the different proposed estimators.

\section{\textbf{MLE of the PDF and the CDF}}

\noindent Suppose $X_{1},\ldots,X_{n}$ is a random sample from the
Topp-Leone distribution given by (1) and (2). The MLE of $\alpha,$ denoted
by $\widehat{\alpha}$, is given by%
\begin{equation*}
\widehat{\alpha}=\frac{-n}{\dsum \limits_{i=1}^{n}\ln\left(
2x_{i}-x_{i}^{2}\right) }.
\end{equation*}
Therefore, using the invariance property of MLE, we obtain the MLEs of the
PDF (1) and the CDF (2) as follows%
\begin{equation*}
\widehat{f}\left( x\right) =\widehat{\alpha}\left( 2-2x\right) \left(
2x-x^{2}\right) ^{\widehat{\alpha}-1}
\end{equation*}
and%
\begin{equation*}
\widehat{F}\left( x\right) =\left( 2x-x^{2}\right) ^{\widehat{\alpha}},
\end{equation*}
respectively for $0<x<1$.\medskip\vspace{0.02in}

\noindent To calculate $E\left( \left[ \widehat{f}\left( x\right) \right]
^{r}\right) $ and $E\left( \left[ \widehat{F}\left( x\right) \right]
^{r}\right) ,$ we need to.find the PDF of random variable $\widehat{\alpha}$
Let $Z_{i}=-\ln\left( 2X_{i}-X_{i}^{2}\right) $, $i=1,\ldots,n$ and $T=\dsum
\limits_{i=1}^{n}Z_{i}.$ Then, the distribution of $Z_{1}$ is exponential
with PDF given by%
\begin{equation*}
f_{Z_{1}}\left( z\right) =\alpha e^{-\alpha z},\text{ for }z>0\text{,}
\end{equation*}
and then $T$ is a gamma random variable with the PDF given by%
\begin{equation*}
f_{T}\left( t\right) =\frac{\alpha^{n}}{\Gamma\left( n\right) }%
t^{n-1}e^{-\alpha t},\text{ for }t>0\text{.}
\end{equation*}
Therefore, using some elementary algebra, the PDF of $\widehat{\alpha}=S=n/T$
is%
\begin{equation}
f_{S}\left( s\right) =\dfrac{n^{n}\alpha^{n}e^{-\frac{n\alpha}{s}}}{%
\Gamma\left( n\right) s^{n+1}},\text{ for }s>0\text{.}
\end{equation}
\medskip\medskip In the following Theorem, we give $E\left( \left[ \widehat{f%
}\left( x\right) \right] ^{r}\right) $ and $E\left( \left[ \widehat{F}\left(
x\right) \right] ^{r}\right) .$\vspace{0.02in}

\begin{theorem}
We have%
\begin{equation*}
E\left( \left[ \widehat{f}\left( x\right) \right] ^{r}\right) =\frac{2\left(
n\alpha\right) ^{\frac{r+n}{2}}\left( 2-2x\right) ^{r}}{\Gamma\left(
n\right) \left( 2x-x^{2}\right) ^{r}}\left( \frac{-1}{r\ln\left(
2x-x^{2}\right) }\right) ^{\frac{r-n}{2}}K_{r-n}\left( 2\sqrt{-n\alpha
r\ln\left( 2x-x^{2}\right) }\right)
\end{equation*}
\medskip and%
\begin{equation*}
E\left( \left[ \widehat{F}\left( x\right) \right] ^{r}\right) =\frac{2\left(
n\alpha\right) ^{\frac{n}{2}}}{\Gamma\left( n\right) }\left( \frac{-1}{%
r\ln\left( 2x-x^{2}\right) }\right) ^{-\frac{n}{2}}K_{-n}\left( 2\sqrt{%
-n\alpha r\ln\left( 2x-x^{2}\right) }\right) ,
\end{equation*}
where $K_{v}\left( \cdot\right) $ is the modified Bessel function of the
second kind of order $v$.\medskip\smallskip
\end{theorem}

\begin{proof}
From Equation (3), we can write%
\begin{align*}
E\left( \left[ \widehat{f}\left( x\right) \right] ^{r}\right) &
=\int_{0}^{\infty}\left[ s\left( 2-2x\right) \left( 2x-x^{2}\right) ^{s-1}%
\right] ^{r}\frac{n^{n}\alpha^{n}e^{-\frac{n\alpha}{s}}}{\Gamma\left(
n\right) s^{n+1}}ds \\
& =\frac{n^{n}\alpha^{n}\left( 2-2x\right) ^{r}}{\Gamma\left( n\right)
\left( 2x-x^{2}\right) ^{r}}\int_{0}^{\infty}s^{r-n-1}\left( 2x-x^{2}\right)
^{rs}e^{-\frac{n\alpha}{s}}ds \\
& =\frac{n^{n}\alpha^{n}\left( 2-2x\right) ^{r}}{\Gamma\left( n\right)
\left( 2x-x^{2}\right) ^{r}}\int_{0}^{\infty}s^{r-n-1}e^{rs\ln\left(
2x-x^{2}\right) }e^{-\frac{n\alpha}{s}}ds.
\end{align*}
Using Formula (3.471.9) in Gradshteyn and Ryzhik (2000), we obtain%
\begin{equation*}
\int_{0}^{\infty}s^{r-n-1}e^{rs\ln\left( 2x-x^{2}\right) }e^{-\frac{n\alpha 
}{s}}ds=2\left( \frac{-n\alpha}{r\ln\left( 2x-x^{2}\right) }\right) ^{\frac{%
r-n}{2}}K_{r-n}\left( 2\sqrt{-n\alpha r\ln\left( 2x-x^{2}\right) }\right) .
\end{equation*}
\medskip\medskip Therefore%
\begin{equation*}
E\left( \left[ \widehat{f}\left( x\right) \right] ^{r}\right) =\frac{2\left(
n\alpha\right) ^{\frac{r+n}{2}}\left( 2-2x\right) ^{r}}{\Gamma\left(
n\right) \left( 2x-x^{2}\right) ^{r}}\left( \frac{-1}{r\ln\left(
2x-x^{2}\right) }\right) ^{\frac{r-n}{2}}K_{r-n}\left( 2\sqrt{-n\alpha
r\ln\left( 2x-x^{2}\right) }\right) .
\end{equation*}
In a similar manner, one can obtain $E\left( \left[ \widehat{F}\left(
x\right) \right] ^{r}\right) \ $and then the proof of Theorem 2.1 is
completed.\medskip
\end{proof}

\begin{remark}
From Theorem 2.1, we observe$\ $(when $r=1$) that the estimators $\widehat{f}%
\left( x\right) $ and $\widehat{F}\left( x\right) $ are biased for $f(x)$
and $F(x)$, respectively.\bigskip\medskip
\end{remark}

\noindent The MSEs of $\widehat{f}\left( x\right) $ and $\widehat{F}\left(
x\right) $ are given in the following Theorem.\smallskip

\begin{theorem}
The MSEs of $\widehat{f}\left( x\right) $ and $\widehat{F}\left( x\right) $
are respectively given by%
\begin{align*}
MSE\left( \widehat{f}\left( x\right) \right) & =\frac{2\left( n\alpha\right)
^{\frac{2+n}{2}}\left( 2-2x\right) ^{2}}{\Gamma\left( n\right) \left(
2x-x^{2}\right) ^{2}}\left( \frac{-1}{2\ln\left( 2x-x^{2}\right) }\right) ^{%
\frac{2-n}{2}} \\[0.05in]
& \times K_{2-n}\left( 2\sqrt{-2n\alpha\ln\left( 2x-x^{2}\right) }\right) -%
\frac{8\left( n\alpha\right) ^{\frac{1+n}{2}}\left( 1-x\right) f\left(
x\right) }{\Gamma\left( n\right) \left( 2x-x^{2}\right) } \\
& \times\left( \frac{-1}{\ln\left( 2x-x^{2}\right) }\right) ^{\frac {1-n}{2}%
}K_{1-n}\left( 2\sqrt{-n\alpha\ln\left( 2x-x^{2}\right) }\right)
+f^{2}\left( x\right) ,
\end{align*}
and%
\begin{align*}
MSE\left( \widehat{F}\left( x\right) \right) & =\frac{2\left( n\alpha\right)
^{\frac{n}{2}}}{\Gamma\left( n\right) }\left( -\frac {1}{2\ln\left(
2x-x^{2}\right) }\right) ^{-\frac{n}{2}}K_{-n}\left( 2\sqrt{%
-2n\alpha\ln\left( 2x-x^{2}\right) }\right) \\[0.06in]
& -\frac{4\left( n\alpha\right) ^{n/2}F\left( x\right) }{\Gamma\left(
n\right) }\left( -\frac{1}{\ln\left( 2x-x^{2}\right) }\right) ^{-n/2} \\
& \times K_{-n}\left( 2\sqrt{-n\alpha\ln\left( 2x-x^{2}\right) }\right)
+F^{2}\left( x\right) .
\end{align*}
\medskip
\end{theorem}

\begin{proof}
We have%
\begin{equation*}
MSE\left( \widehat{f}\left( x\right) \right) =E\left( \left[ \widehat{f}%
\left( x\right) \right] ^{2}\right) -2f\left( x\right) E\left( \widehat{f}%
\left( x\right) \right) +f^{2}\left( x\right) ,
\end{equation*}
then, by setting $r=1$ and $r=2$ in Theorem 2.1, we obtain $MSE\left( 
\widehat{f}\left( x\right) \right) $. In a similar manner, we can find $%
MSE\left( \widehat{F}\left( x\right) \right) $.
\end{proof}

\section{\textbf{UMVUE of the PDF and the CDF}}

\noindent In this section, the UMVUEs of the PDF and the CDF of the
Topp-Leone distribution are derived. The MSEs of these estimators are also
obtained.\smallskip

\noindent Suppose $X_{1},\ldots,X_{n}$ is a random sample from the
Topp-Leone distribution. Then $T=-\dsum \limits_{i=1}^{n}\ln\left(
2X_{i}-X_{i}^{2}\right) $ is a complete sufficient statistic for $\alpha.$
The UMVUE of $f(x)$,\ denoted by $f^{\ast}\left( t\right) ,$ is given by
Lehmann-Scheff\'{e} theorem%
\begin{equation*}
E\left( f^{\ast}\left( T\right) \right) =\int
f_{X_{1}|T}(x_{1}|t)f_{T}\left( t\right) dt=\int
f_{X_{1},T}(x_{1},t)dt=f_{X_{1}}(x_{1}),
\end{equation*}
where $f_{X_{1}|T}(x_{1}|t)=f^{\ast}(t)$ is the conditional PDF of $X_{1}$
given $T$ and $f_{X_{1},T}(x_{1},t)$ is the joint PDF of $X_{1}$ and $T$.
Then, we need the following Lemma to find $f_{X_{1}|T}(x_{1}|t)$.

\begin{lemma}
The conditional PDF of $X_{1}$ given $T$ is%
\begin{equation*}
f_{X_{1}|T}(x|t)=\frac{\left( n-1\right) \left( 2-2x\right) \left(
t+\ln\left( 2x-x^{2}\right) \right) ^{n-2}}{\left( 2x-x^{2}\right) t^{n-1}},%
\text{\ }-\ln\left( 2x-x^{2}\right) <t<\infty
\end{equation*}
\end{lemma}

\begin{proof}
We have $T=\dsum \limits_{i=1}^{n}Z_{i}$ is a gamma random variable with the
PDF given by (3). Therefore, using some elementary algebra, the conditional
PDF of $Z_{1}$ given $T$ is\textbf{\ }%
\begin{equation*}
h_{Z_{1}|T}(z|t)=\frac{\left( n-1\right) \left( t-z\right) ^{n-2}}{t^{n-1}},%
\text{ }0<z<t.
\end{equation*}
Then%
\begin{equation*}
f_{X_{1}|T}(x_{1}|t)=\frac{2-2x}{2x-x^{2}}h_{Z_{1}|T}(-\ln\left(
2x-x^{2}\right) |t),\text{ }-\ln\left( 2x-x^{2}\right) <t<\infty,
\end{equation*}
and the proof is completed.\bigskip
\end{proof}

\begin{theorem}
Let $T=t$ be given. Then, the UMVUEs of $f\left( x\right) $ and $F\left(
x\right) $ are respectively given, for $-\ln\left( 2x-x^{2}\right)
<t<\infty, $ by%
\begin{equation*}
\widetilde{f}\left( x\right) =\frac{\left( n-1\right) \left( 2-2x\right)
\left( t+\ln\left( 2x-x^{2}\right) \right) ^{n-2}}{\left( 2x-x^{2}\right)
t^{n-1}}
\end{equation*}
and%
\begin{equation*}
\widetilde{F}\left( x\right) =\left( \frac{t+\ln\left( 2x-x^{2}\right) }{t}%
\right) ^{n-1}.
\end{equation*}
\end{theorem}

\bigskip

\begin{proof}
From Lemma 3.1, we see immediately that $\widetilde{f}\left( x\right) $ is
the UMVUE of $f\left( x\right) $. Also, we obtain $\widetilde{F}\left(
x\right) $ by integrating $\widetilde{f}\left( x\right) .$\bigskip
\end{proof}

\noindent The MSEs of $\widetilde{f}\left( x\right) $ and $\widetilde{F}%
\left( x\right) $ are given in the following Theorem.\medskip

\begin{theorem}
The MSEs of $f\left( x\right) $ and $F\left( x\right) $ are respectively
given by%
\begin{equation*}
MSE\left( \widetilde{f}\left( x\right) \right) =\frac{A^{2}}{\Gamma\left(
n\right) }\sum_{i=1}^{2n-4}\binom{2n-4}{i}b^{i}\alpha^{i+3}\Gamma\left(
n-i-2,-\alpha b\right) -f^{2}\left( x\right)
\end{equation*}
and%
\begin{equation*}
MSE\left( \widetilde{F}\left( x\right) \right) =\frac{1}{\Gamma\left(
n\right) }\sum_{i=1}^{2n-2}\binom{2n-2}{i}b^{i}\alpha^{i+1}\Gamma\left(
n-i,-\alpha b\right) -F^{2}\left( x\right) ,
\end{equation*}
where $A=\frac{\left( n-1\right) \left( 2-2x\right) }{2x-x^{2}},$ $%
b=\ln\left( 2x-x^{2}\right) $ and $\Gamma\left( s,x\right)
=\int_{x}^{\infty}t^{s-1}e^{-t}dt$ is the complementary incomplete gamma
function.\medskip
\end{theorem}

\begin{proof}
We have%
\begin{equation*}
\widetilde{f}\left( x\right) =\frac{A\left( t+b\right) ^{n-2}}{t^{n-1}},
\end{equation*}
and%
\begin{equation*}
MSE\left( \widetilde{f}\left( x\right) \right) =E\left( \left[ \widetilde{f}%
\left( x\right) \right] ^{2}\right) -f^{2}\left( x\right) ,
\end{equation*}
where, from Equation (3),%
\begin{align*}
E\left( \left[ \widetilde{f}\left( x\right) \right] ^{2}\right) & =\frac{%
\alpha^{n}A^{2}}{\Gamma\left( n\right) }\int_{-b}^{\infty}\left[ \frac{%
\left( t+b\right) ^{n-2}}{t^{n-1}}\right] ^{2}t^{n-1}e^{-\alpha t}dt \\
& =\frac{\alpha^{n}A^{2}}{\Gamma\left( n\right) }\int_{-b}^{\infty}\frac{%
\left( t+b\right) ^{2n-4}}{t^{2n-2}}t^{n-1}e^{-\alpha t}dt \\
& =\frac{\alpha^{n}A^{2}}{\Gamma\left( n\right) }\int_{-b}^{\infty}\frac{%
\left( t+b\right) ^{2n-4}}{t^{2n-4}}t^{n-3}e^{-\alpha t}dt \\
& =\frac{\alpha^{n}A^{2}}{\Gamma\left( n\right) }\int_{-b}^{\infty}\left( 
\frac{t+b}{t}\right) ^{2n-4}t^{n-3}e^{-\alpha t}dt \\
& =\frac{\alpha^{n}A^{2}}{\Gamma\left( n\right) }\int_{-b}^{\infty}\left( 1+%
\frac{b}{t}\right) ^{2n-4}t^{n-3}e^{-\alpha t}dt.
\end{align*}
Using%
\begin{equation*}
\left( 1+\frac{b}{t}\right) ^{2n-4}=\sum_{i=1}^{2n-4}\binom{2n-4}{i}\left( 
\frac{b}{t}\right) ^{i},
\end{equation*}
we obtain%
\begin{align*}
E\left( \left[ \widetilde{f}\left( x\right) \right] ^{2}\right) & =\frac{%
\alpha^{n}A^{2}}{\Gamma\left( n\right) }\int_{-b}^{\infty}\left[
\sum_{i=1}^{2n-4}\binom{2n-4}{i}\left( \frac{b}{t}\right) ^{i}\right]
t^{n-3}e^{-\alpha t}dt \\
& =\frac{\alpha^{n}A^{2}}{\Gamma\left( n\right) }\sum_{i=1}^{2n-4}\binom{2n-4%
}{i}b^{i}\int_{-b}^{\infty}t^{n-i-3}e^{-\alpha t}dt.
\end{align*}
The change of variables $u=\alpha t,$\ yields%
\begin{align*}
E\left( \left[ \widetilde{f}\left( x\right) \right] ^{2}\right) & =\frac{%
\alpha^{n}A^{2}}{\Gamma\left( n\right) }\sum_{i=1}^{2n-4}\binom {2n-4}{i}%
\frac{b^{i}}{\alpha^{n-i-3}}\int_{-\alpha b}^{\infty}u^{n-i-3}e^{-u}du \\
& =\frac{A^{2}}{\Gamma\left( n\right) }\sum_{i=1}^{2n-4}\binom{2n-4}{i}%
b^{i}\alpha^{i+3}\Gamma\left( n-i-2,-\alpha b\right) .
\end{align*}
In a similar manner, we can find $MSE\left( \widetilde{F}\left( x\right)
\right) $.
\end{proof}

\section{\textbf{PCE of the PDF and the CDF}}

\noindent The PCE was first introduced by Kao (1958, 1959)\ and used it when
some probability distribution has the quartile function in a closed-form
expression. In this section, the PCEs of the PDF and the CDF of the
Topp-Leone distribution are derived. Let $X_{1},\ldots,X_{n}$ be a random
sample of size $n$ from the Topp-Leone distribution in (2), and let $%
X_{(1)},\ldots,X_{(n)}$ be the corresponding order statistics. Then the PCE
of $\alpha$, denoted by $\widetilde{\alpha}_{PCE},$ is obtained by minimizing%
\begin{equation*}
\sum_{j=1}^{n}\left( p_{i}^{1/\alpha}-\left( 2x_{(i)}-x_{(i)}^{2}\right)
\right) ^{2},
\end{equation*}
where $p_{i}=\frac{i}{n+1}$. So, the PCEs of the PDF and the CDF are
respectively given, for $0<x<1,$ by%
\begin{equation*}
\widetilde{f}_{PCE}\left( x\right) =\widetilde{\alpha}_{PCE}\left(
2-2x\right) \left( 2x-x^{2}\right) ^{\widetilde{\alpha}_{PCE}-1},
\end{equation*}
and%
\begin{equation*}
\widetilde{F}_{PCE}\left( x\right) =\left( 2x-x^{2}\right) ^{\widetilde{%
\alpha}_{PCE}}.
\end{equation*}
Since it's difficult to find the MSEs of these estimators analytically, we
shall calculate them by simulation.

\section{\textbf{LSE and WLSE of the PDF and the CDF}}

The LSE and WLSE were originally proposed by Swain et al. (1988) to estimate
the parameters of Beta distributions. In this section, we use the same
methods for the Topp-Leone distribution. Let $X_{(1)},\ldots,X_{(n)}$ be the
order statistics of a size $n$ random sample from a distribution with CDF $%
G\left( \cdot\right) $. It is well known that%
\begin{equation*}
E\left( G(X_{(i)}\right) =\frac{i}{n+1}\text{ and\ }Var\left(
G(X_{(i)}\right) =\frac{i\left( n-i+1\right) }{\left( n+1\right) ^{2}\left(
n+2\right) }.
\end{equation*}
Using the expectations and the variances, two variants of the least squares
method follow.\bigskip

\subsection{\textit{LSE of the PDF and the CDF}}

The LSE of the unknown parameter is obtained by minimizing the function%
\begin{equation*}
\sum_{i=1}^{n}\left( F\left( X_{\left( i\right) }\right) -\frac{i}{n+1}%
\right) ^{2},
\end{equation*}
with respect to the unknown parameters. Therefore, for the Topp-Leone
distribution the LSE of $\alpha,$\textit{\ }denoted by $\widetilde{\alpha }%
_{LS},$is obtained by minimizing%
\begin{equation*}
\sum_{i=1}^{n}\left( \left( 2x_{\left( i\right) }-x_{\left( i\right)
}^{2}\right) ^{\alpha}-\frac{i}{n+1}\right) ^{2}.
\end{equation*}
Then, the LSEs of the PDF and the CDF are respectively given, for $0<x<1,$ by%
\begin{equation*}
\widetilde{f}_{LS}\left( x\right) =\widetilde{\alpha}_{LS}\left( 2-2x\right)
\left( 2x-x^{2}\right) ^{\widetilde{\alpha}_{LS}-1},
\end{equation*}
and%
\begin{equation*}
\widetilde{F}_{LS}\left( x\right) =\left( 2x-x^{2}\right) ^{\widetilde{\alpha%
}_{LS}}.
\end{equation*}
Since it's difficult to find the MSEs of these estimators analytically, we
shall calculate them by simulation.

\subsection{\textit{WLSE of the PDF and the CDF}}

The WLSE of the unknown parameter is obtained by minimizing%
\begin{equation*}
\sum_{i=1}^{n}w_{i}\left( F\left( X_{\left( i\right) }\right) -\frac{i}{n+1}%
\right) ^{2}
\end{equation*}%
with respect to the unknown parameters, where%
\begin{equation*}
w_{i}=\frac{1}{Var\left( F\left( X_{\left( i\right) }\right) \right) }=\frac{%
\left( n+2\right) \left( n+1\right) ^{2}}{i\left( n-i+1\right) }.
\end{equation*}%
Therefore, for the Topp-Leone distribution the WLSE of $\alpha ,$\textit{\ }%
denoted by $\widetilde{\alpha }_{WLS},$ is obtained by minimizing%
\begin{equation*}
\sum_{i=1}^{n}w_{i}\left( \left( 2x_{\left( i\right) }-x_{\left( i\right)
}^{2}\right) ^{\alpha }-\frac{i}{n+1}\right) ^{2}.
\end{equation*}%
Then, the WLS estimators of the PDF and the CDF are respectively given, for $%
0<x<1,$ by%
\begin{equation*}
\widetilde{f}_{WLS}\left( x\right) =\widetilde{\alpha }_{WLS}\left(
2-2x\right) \left( 2x-x^{2}\right) ^{\widetilde{\alpha }_{WLS}-1}
\end{equation*}%
and%
\begin{equation*}
\widetilde{F}_{WLS}\left( x\right) =\left( 2x-x^{2}\right) ^{\widetilde{%
\alpha }_{WLS}}.
\end{equation*}%
Since it's difficult to find the MSEs of these estimators analytically, we
shall calculate them by simulation.

\section{\textbf{Simulation study}}

In this section, by means of the statistical software \textbf{R}, a
simulation study is carried out to compare the different proposed
estimators: the MLE, the UMVUE, the PCE, the LSE and the WLSE of the PDF and
the CDF. This comparison is based on the MSEs. To this end, we generate 1000
independent replicates of sizes $n=10$, $20$, $50$ and $100$ from the
Topp-Leone distribution with $\alpha=0.5,$ $1.0$, $2.0$ and $3.0$. The
Topp-Leone random number generation is performed using the inversion method: 
$X=1-\sqrt {1-U^{1/\alpha}},$ where $U$ is a standard uniform random
variable. The results are summarized in Table 1. From this table, we observe
that the MLEs of the PDF and the CDF have the smallest MSEs. Then, the MLE
perform better than the others estimators in all the cases considered in
terms of MSEs. Also, we observe that the MSEs for each estimator decrease
with increasing sample size as expected.

\bigskip 
%TCIMACRO{\TeXButton{B}{\begin{table}[h] \centering}}%
%BeginExpansion
\begin{table}[h] \centering%
%EndExpansion
\caption{MSEs of the PDF estimators and the CDF estimators (MSEs of the CDF
estimators are given in brackets). }%
\begin{tabular}{|l|c|c|c|c|c|}
\hline
$n$ & Method & $\alpha =0.5$ & $\alpha =1.0$ & $\alpha =2.0$ & $\alpha =3.0$
\\ \hline
$10$ & \multicolumn{1}{|l|}{MLE} & 0.0702$\left[ 0.0892\right] $ & 0.0837$%
\left[ 0.0112\right] $ & 0.0238$\left[ 0.1021\right] $ & 0.0511$\left[ 0.0464%
\right] $ \\ 
& \multicolumn{1}{|l|}{UMVUE} & 0.0921$\left[ 0.1193\right] $ & 0.1021$\left[
0.0920\right] $ & 0.0436$\left[ 0.1355\right] $ & 0.0518$\left[ 0.0678\right]
$ \\ 
& \multicolumn{1}{|l|}{PCE} & 0.1296$\left[ 0.1247\right] $ & 0.1313$\left[
0.1047\right] $ & 0.0982$\left[ 0.1513\right] $ & 0.1099$\left[ 0.0897\right]
$ \\ 
& \multicolumn{1}{|l|}{LSE} & 0.1941$\left[ 0.2158\right] $ & 0.1967$\left[
0.2001\right] $ & 0.1869$\left[ 0.1901\right] $ & 0.2010$\left[ 0.1649\right]
$ \\ 
& \multicolumn{1}{|l|}{WLSE} & 0.1604$\left[ 0.1847\right] $ & 0.1741$\left[
0.1699\right] $ & 0.1107$\left[ 0.1318\right] $ & 0.1334$\left[ 0.1121\right]
$ \\ \hline
$20$ & \multicolumn{1}{|l|}{MLE} & 0.0314$\left[ 0.1701\right] $ & 0.0716$%
\left[ 0.0053\right] $ & 0.020$\left[ 0.0921\right] $ & 0.0341$\left[ 0.0344%
\right] $ \\ 
& \multicolumn{1}{|l|}{UMVUE} & 0.0422$\left[ 0.0985\right] $ & 0.0921$\left[
0.1192\right] $ & 0.0398$\left[ 0.1040\right] $ & 0.0403$\left[ 0.0516\right]
$ \\ 
& \multicolumn{1}{|l|}{PCE} & 0.0956$\left[ 0.1003\right] $ & 0.1296$\left[
0.1247\right] $ & 0.0772$\left[ 0.1106\right] $ & 0.0885$\left[ 0.0662\right]
$ \\ 
& \multicolumn{1}{|l|}{LSE} & 0.1540$\left[ 0.1816\right] $ & 0.1941$\left[
0.2158\right] $ & 0.1552$\left[ 0.1661\right] $ & 0.1819$\left[ 0.1398\right]
$ \\ 
& \multicolumn{1}{|l|}{WLSE} & 0.1297$\left[ 0.1147\right] $ & 0.1601$\left[
0.1841\right] $ & 0.1079$\left[ 0.1282\right] $ & 0.1223$\left[ 0.0997\right]
$ \\ \hline
$50$ & \multicolumn{1}{|l|}{MLE} & 0.0095$\left[ 0.0080\right] $ & 0.0102$%
\left[ 0.0026\right] $ & 0.0010$\left[ 0.0352\right] $ & 0.0028$\left[ 0.0039%
\right] $ \\ 
& \multicolumn{1}{|l|}{UMVUE} & 0.0193$\left[ 0.0248\right] $ & 0.0467$\left[
0.0788\right] $ & 0.0035$\left[ 0.0502\right] $ & 0.0177$\left[ 0.0207\right]
$ \\ 
& \multicolumn{1}{|l|}{PCE} & 0.04562$\left[ 0.0803\right] $ & 0.0816$\left[
0.1098\right] $ & 0.0225$\left[ 0.0651\right] $ & 0.0412$\left[ 0.0388\right]
$ \\ 
& \multicolumn{1}{|l|}{LSE} & 0.10410$\left[ 0.1558\right] $ & 0.1511$\left[
0.1772\right] $ & 0.0998$\left[ 0.1136\right] $ & 0.1196$\left[ 0.1006\right]
$ \\ 
& \multicolumn{1}{|l|}{WLSE} & 0.0974$\left[ 0.10471\right] $ & 0.1153$\left[
0.1380\right] $ & 0.0296$\left[ 0.0441\right] $ & 0.0757$\left[ 0.0567\right]
$ \\ \hline
$100$ & \multicolumn{1}{|l|}{MLE} & 0.00021$\left[ 0.00027\right] $ & 
\multicolumn{1}{|l|}{0.00049$\left[ 0.0016\right] $} & \multicolumn{1}{|l|}{
0.00016$\left[ 0.0017\right] $} & \multicolumn{1}{|l|}{0.00018$\left[ 0.00047%
\right] $} \\ 
& \multicolumn{1}{|l|}{UMVUE} & 0.0022$\left[ 0.0012\right] $ & 0.0050$\left[
0.0042\right] $ & 0.00077$\left[ 0.0123\right] $ & 0.0092$\left[ 0.0084%
\right] $ \\ 
& \multicolumn{1}{|l|}{PCE} & 0.0033$\left[ 0.0016\right] $ & 0.00708$\left[
0.0056\right] $ & 0.0087$\left[ 0.0209\right] $ & 0.0104$\left[ 0.0196\right]
$ \\ 
& \multicolumn{1}{|l|}{LSE} & 0.0727$\left[ 0.0801\right] $ & 0.0930$\left[
0.1004\right] $ & 0.0299$\left[ 0.0674\right] $ & 0.0789$\left[ 0.0817\right]
$ \\ 
& \multicolumn{1}{|l|}{WLSE} & 0.0131$\left[ 0.0015\right] $ & 0.0431$\left[
0.086\right] $ & 0.0099$\left[ 0.0376\right] $ & 0.0327$\left[ 0.0229\right] 
$ \\ \hline
\end{tabular}%
%TCIMACRO{\TeXButton{E}{\end{table}}}%
%BeginExpansion
\end{table}%
%EndExpansion

\end{document}